\documentclass[11pt]{article}
\usepackage{amsmath,amssymb,amsthm}
\usepackage[nocompress]{cite}
\usepackage{geometry}
\usepackage{enumitem}
\usepackage{tikz}
\usepackage{hyperref}
\usepackage{fullpage}

\usepackage{graphicx} 

\usepackage{xcolor}
\usepackage{amsmath,amssymb,amsthm}
\theoremstyle{plain} 
\newtheorem{theorem}{Theorem}
\newtheorem{lemma}[theorem]{Lemma} 

\theoremstyle{definition} 
\newtheorem{example}[theorem]{Example} 

\newcommand{\lms}{\ensuremath\mathsf{leftmatchstates}}
\newcommand{\rms}{\ensuremath\mathsf{rightmatchstates}}
\newcommand{\prefixmatch}{\ensuremath\mathsf{prefixmatch}}
\newcommand{\suffixmatch}{\ensuremath\mathsf{suffixmatch}}
\newcommand{\prefixset}{\ensuremath\mathsf{prefixset}}
\newcommand{\suffixset}{\ensuremath\mathsf{suffixset}}
\newcommand{\inttraversal}{\ensuremath\mathsf{IntTraversal}}
\newcommand{\rootnode}{\ensuremath\mathit{root}}
\newcommand{\occ}{\ensuremath\mathit{occ}}

\title{Regular Expression Matching with Ordered Linear Backreferences}
\author{Philip Bille \\\texttt{phbi@dtu.dk} \and Inge Li G{\o}rtz \\\texttt{inge@dtu.dk} \and Rikke Schjeldrup Jessen \\\texttt{rscje@dtu.dk}}

\date{}

\begin{document}

\maketitle

\begin{abstract}
A regular expression with backreferences (rewb) specifies a set of strings formed by characters combined with concatenation, union, star operators, and backreferences. A backreference consists of a capturing group $(\cdot)_i$ and one or more references $\backslash i$. The substring matched by the reference must match the same substring that was matched by the corresponding capturing group. 
We focus on a subclass, which we call ordered linear rewbs where the references and capturing groups occur only outside star and union operators. Furthermore, we require that all references to a capturing group occur before the next capturing group in the rewb.
Rewb matching is a basic tool in computer science for searching and processing text, is supported in most modern programming languages, and is used across a wide range of applications. 

In full generality, rewb matching is NP-complete, but efficient solutions exist for various subclasses. For linear rewbs with one capturing group and $k$ references, Uezato~[CPM 2026] gave an $O(k n^2 m^2)$ time and $O(n^2m^2)$ space algorithm, where $m$ is the length of the regular expression $R$ and $n$ is the length of the string $Q$. For the special case of $k=1$, Nogami and Terauchi~[MFCS 2025] gave an $O(n^2m^2)$ time and $O(n+ m^2)$ space algorithm. On the other hand, Nogami, Nakamura, and Terauchi~[arXiv 2026] gave a conditional lower bound, showing that we cannot solve the problem in $O(n^{2-\epsilon} \mathrm{poly}(m))$ for any $\epsilon > 0$ assuming the orthogonal vector hypothesis. Our main result is a new improved algorithm that runs in $O(n^2m+n^2k\log n)$ time and uses $O(nm)$ space, and handles the more general case of ordered linear rewbs with multiple capturing groups. 

To achieve our results, we combine suffix-tree enumeration of candidate strings for the capturing groups with a new interval-tree data structure. The interval tree stores state-sets of Thompson's finite automaton for dyadic substrings of the input and lets us efficiently propagate feasible match endpoints across multiple references and capturing groups.

\end{abstract}


\section{Introduction}
A regular expression with backreferences (rewb) specifies a set of strings formed by characters combined with concatenation ($\odot$), union ($|$), star ($^*$) operator, and backreferences. A backreference consists of a numbered \emph{capturing group} $(e)_i$ and a corresponding numbered reference $\backslash i$. The substring matched by the reference has to be identical to the substring matched by the corresponding capturing group. For example, the language of the rewb $R=c(a^*)_1b \backslash1$ (by convention, we often omit the $\odot$ symbol when writing expressions) is $L(R)=\{cb,caba,caabaa,\ldots\}$.

In this paper, we focus on a specific subclass of rewbs, which we call \emph{ordered linear} rewbs. A rewb is \emph{linear} if neither the star nor the union operator is used on a subexpression containing a capturing group or a backreference. For a linear rewb with  a single capturing group, this means that the expression is of the form $R=e_0(e)_1e_1\backslash1~e_2\backslash1~e_3\ldots e_{k}\backslash1~e_{k+1}$, where $e$ and $e_i$ are standard regular expressions without backreferences.
For example, $a(b^*)_1\backslash1~b$ is linear, while $a(b)_1^*\backslash1~b$ and $a(b)_1(\backslash1|b)$ are not.
We call a rewb a \emph{single linear} rewb if it is a linear rewb with only one capturing group. We call a rewb an \emph{ordered linear} rewb if it is a concatenation of single linear rewbs, i.e., it is of the form $E_0E_1\ldots E_t$ where $E_i$ for $i=0\ldots t$ are single linear rewbs. Note that a single linear rewb is a subclass of ordered linear rewbs. 
This class of expressions captures many practical uses of rewb~\cite{NT2025}.


Rewb matching is a basic tool in computer science for searching and processing text and is supported in most modern programming languages. Rewb matching appears across a wide range of applications such as network intrusion detection~\cite{NamjoshiNarlikar2010},
graph-database querying~\cite{Schmid2020}, software testing and symbolic execution~\cite{LoringMitchellKinder2019}, and string-constraint solving~\cite{ChenEtAl2022}.

In full generality, rewb matching is NP-complete~\cite{AHO1990255} and upper and lower bounds have been studied for various subclasses of rewb~\cite{uezato2026, NT2025, LARSEN1998, FREYDENBERGER20191, HHLVG2026, KU2026, nogami2026hardness}. 
 Let 
 $R$ be a single linear rewb of length $m$ with $k\geq 1$ references 
 and let $Q$ be a string of length $n$. For $k=1$, (i.e., $R = e_0(e)_1e_1\backslash1 e_2$), Nogami and Terauchi~\cite{NT2025} gave an $O(n^2m^2)$ time and $O(n+ m^2)$ space algorithm. For general $k$, Uezato~\cite{uezato2026} gave an $O(k n^2 m^2)$ time and $O(n^2m^2)$ space algorithm. On the other hand, Nogami, Nakamura, and Terauchi~\cite{nogami2026hardness} gave a conditional lower bound, showing that we cannot solve the problem in $O(n^{2-\epsilon}\mathrm{poly}(m))$ time for any $\epsilon>0$, assuming the orthogonal vectors hypothesis. Thus, the exponent on $n$ is conditionally optimal for a polynomial dependency on $m$. A recent result by Kumabe and Uezato~\cite{KU2026} shows that it is possible to obtain a subquadratic time dependency on $n$ at the cost of a multiplicative exponential dependency on $m$.

\subsection{Setup and Results}
Given a single linear rewb $R$ and a string $Q$, the \emph{single linear rewb matching problem} is to decide whether $Q$ is one of the strings specified by $R$. 
Likewise, given an ordered linear rewb $R$ and a string $Q$, the \emph{ordered linear rewb matching problem} is to decide whether $Q$ is one of the strings specified by $R$. 
Let $m$ and $n$ be the length (number of symbols) of the rewb $R$ and the string $Q$, respectively. All bounds presented below and in the paper hold for a standard unit-cost RAM with $w$-bit machine words supporting standard arithmetic and logical operations on machine words. This means that the algorithms can be implemented directly in standard imperative programming languages such as C or C++. An index into $R$ and $Q$ can be stored in a single machine word and thus $w \geq \log (n + m)$. 
Our main result is the following. 

\begin{theorem}\label{thm:ordered_linear}
    Given an ordered linear rewb of length $m$ with $k$ references and a string of length $n$, we can solve the \emph{ordered linear rewb matching} problem in $O(n^2m+n^2k\log n)$ time and $O(nm)$ space. 
\end{theorem}
Theorem~\ref{thm:ordered_linear} improves the previous results for the single linear rewb matching problem. Compared to Uezato~\cite{uezato2026},  we improve the space by a factor $nm$. Furthermore,  we improve the time by a factor $km$ when $\log n \leq m$ and with a factor $m^2/\log n$ when $m <\log n< m^2$. Compared to the Nogami and Terauchi~\cite{NT2025} result for $k=1$, we improve the time bound by a factor $m$ when $\log n \leq m$ and $m^2/\log n$ when $m <\log n< m^2$.




\subsection{Technical overview}
Our main tool is an \emph{interval tree}, a complete binary tree $T$ whose nodes represent dyadic intervals of the input string $Q$. For a regular expression $e$ (without backreferences), each node $v$ in $T$ stores, for every position on either side of its midpoint, a bit vector representing the states in a finite automaton (Thompson's NFA construction~\cite{Thompson1968}) that are reachable by the corresponding left or right substring. Any pair of endpoints of a substring is separated at a unique lowest node of the tree, and the substring between them matches $e$ precisely when the corresponding left and right state-sets intersect. Consequently, a tree traversal can process many candidate endpoints simultaneously by taking unions of left state-sets and intersecting them with right state-sets. 

For an expression $R=e_0(e)_1e_1\backslash 1e_2$, we use a suffix tree together with Thompson's NFA for $e$ to enumerate every distinct substring $q$ of $Q$ that matches $e$, together with all occurrences of $q$. From these occurrences and the matches of $e_0$ and $e_2$, we obtain the possible boundaries surrounding~$e_1$. Then, traversing the interval tree for $e_1$ we determine whether some choice gives a match of $R$. We show that an amortized analysis over all $q$ and all tree nodes gives $O(n^2m+n^2\log n)$ time and $O(nm)$ space. For a single linear rewb, we traverse one interval tree for each intervening regular expression and propagate the feasible prefix endpoints from one traversal to the next. Finally, an ordered linear rewb decomposes into a sequence of single linear rewbs, which we process successively while carrying forward the feasible endpoints.

\subsection{Outline}
In Section~\ref{sec:preliminaries} we review notation for strings, regular expressions, automata and suffix trees. Section~\ref{sec:intervaltree} introduces the interval tree data structure. In Section~\ref{sec:simplealgorithm}, we present an algorithm to match the simpler class of single linear rewbs with only one reference. In Section~\ref {sec:generalization1}, we generalize the result to single linear rewbs in general, and in Section~\ref{sec:generalization2} we generalize further to ordered linear rewbs.

\section{Preliminaries}\label{sec:preliminaries}
\subsection{Strings}
A string $Q$ over an alphabet $\Sigma$ is a sequence of characters $Q[0]Q[1]\ldots Q[n-1]$, such that $Q[i]\in \Sigma$ is the $i$'th character of $Q$, and the substring $Q[i,j]$ is the concatenation of characters $Q[i]Q[i+1]\ldots Q[j]$. We also say that $Q[i,i-1]$ is the empty string, which we denote $\varepsilon$.

\subsection{Regular Expressions} 
We consider the set of non-empty regular expressions over an alphabet $\Sigma$, defined recursively as follows. If $\alpha \in \Sigma \cup \{\varepsilon\}$, then $\alpha$ is a regular expression, and if $S$ and $T$ are regular expressions, then so is the \emph{concatenation} ($S \odot T$), the \emph{union} ($S|T$) and the \emph{star} ($S^*$). We often omit the concatenation $\odot$ when writing  regular expressions. 
A \emph{backreference} consists of a numbered capturing group $(e)_i$ and one or more corresponding numbered references $\backslash i$. The substring matched by a reference has to be identical to the substring matched by the corresponding capturing group.
A regular expression with backreferences (rewb) is linear if neither the star nor the union operator is used on a subexpression containing a capturing group or a backreference. For a rewb with a single capturing group this means, that the expression is of the form $R=e_0(e)_1e_1\backslash1~e_2\backslash1~e_3\ldots e_{k}\backslash1~e_{k+1}$, where $e$ and $e_i$ for $i=0\ldots k+1$ are standard regular expressions without backreferences.
For example, the rewb $a(b^*)_1\backslash1~b$ is linear, while $a(b)_1^*\backslash1~b$ and $a(b)_1(\backslash1|b)$ are not.
We call a rewb a \emph{single linear} rewb if it is a linear rewb with only one capturing group. We assume in this paper, that the capturing group is referenced at least once, that is, $k \geq 1$. If there is no reference to the capturing group, the expression can be simplified to a regular expression. We call a rewb an \emph{ordered linear} rewb if it is the concatenation of single linear rewbs, i. e., it is of the form $E_0E_1\ldots E_t$ where $E_i$ for $i=0\ldots t$ are single linear rewbs. Note that a single linear rewb is a subclass of ordered linear rewbs.

The \emph{language} $L(R)$ generated by a regular expression $R$ is defined recursively as follows. If $\alpha \in \Sigma \cup \{\varepsilon\}$, then $L(\alpha)$ is the set containing the single string $\alpha$.  If $S$ and $T$ are regular expressions, then $L(S \odot T) = L(S)\odot L(T)$, that is, any string formed by the concatenation of a string in $L(S)$ with a string in $L(T)$, $L(S | T) = L(S) \cup L(T)$ and $L(S^*) = \bigcup_{i \geq 0} L(S)^i$, where $L(S)^0 = \{\varepsilon\}$ and $L(S)^i = L(S)^{i-1} \odot L(S)$, for $i > 0$. 
The language $L(R)$ of a single linear rewb $R=S_0(T)_1S_1\backslash1~S_2\backslash1~S_3\ldots S_k\backslash1~S_{k+1}$ is defined as follows. Let $S_i$ and $T$ be regular expressions for $i=0\ldots k+1$. Then $L(R)=\bigcup_{t \in L(T)}L(S_0)\odot t\odot L(S_1)\odot t\odot L(S_2)\odot \ldots \odot L(S_k)\odot t\odot L(S_{k+1})$. 
For an ordered linear rewb $R=E_0E_1\ldots E_t$, the language of $R$ is $L(R)=L(E_0)\odot L(E_1)\odot \ldots \odot L(E_t)$.

Given a single linear rewb $R$ and a string $Q$, the \emph{single linear rewb matching problem} is to decide whether $Q$ is one of the strings specified by $R$. 
Likewise, given an ordered linear rewb $R$ and a string $Q$, the \emph{ordered linear rewb matching problem} is to decide whether $Q$ is one of the strings specified by $R$. 

\subsection{Finite Automata}
A \emph{finite automaton} is a tuple $A = (V, E, \Sigma, \Theta, \Phi)$, where $V$ is a set of nodes called \emph{states}, $E \subseteq (V \times V \times (\Sigma \cup \{\varepsilon\}))$ is a set of directed edges between states called \emph{transitions} each labeled by a character from $\Sigma \cup \{\varepsilon\}$, $\Theta \subseteq V$ is a set of \emph{start states}, and $\Phi \subseteq V$ is a set of \emph{accepting states}. In short, $A$ is an edge-labeled directed graph with designated subsets of start and accepting nodes. $A$ is a \emph{deterministic finite automaton} (DFA) if $A$ does not contain any $\varepsilon$-transitions,  all outgoing transitions of any state have different labels, and there is exactly one start state. Otherwise, $A$ is a \emph{nondeterministic finite automaton} (NFA).  For an automaton $A$, we also define the \emph{reverse automaton} $\overleftarrow{A}$, which is the automaton with the same states as $A$, but with the transition $(u,v,a)$ for every transition $(v,u,a)$ in $A$. In addition, the start and accepting states are swapped.

 Given a string $Q$ and a path $p$ in $A$, we say that $p$ and $Q$ match if the concatenation of the labels on the transitions in $p$ is $Q$. Given a state $s$ in $A$ and a character $\alpha \in \Sigma\cup \{\varepsilon\}$ we define the \emph{state-set transition} $\delta_A(s, \alpha)$ to be the set of states reachable from $s$ through paths matching $\alpha$ (note that the paths may include transitions labeled $\varepsilon$). For a set of states $S$ we define $\delta_A(S,\alpha) = \bigcup_{s\in S} \delta_A(s,\alpha)$. We say that $A$ \emph{accepts} a string $Q$ if there is a path from a state in $\Theta$ to a state in $\Phi$ that matches $Q$. Otherwise, $A$ \emph{rejects} $Q$. We can use a sequence of state-set transitions to test acceptance of a string $Q$ of length~$n$ by computing a sequence of state-sets $S_0, \ldots, S_n$, given by $S_0 = \delta_A(\Theta, \varepsilon)$ and $S_i = \delta_A(S_{i-1}, Q[i-1])$, $i=1, \ldots, n$. We have that $\Phi \cap S_n \neq \emptyset$ iff $A$ accepts $Q$.

\begin{figure}[t]
    \centering
    \includegraphics[width=0.7\linewidth]{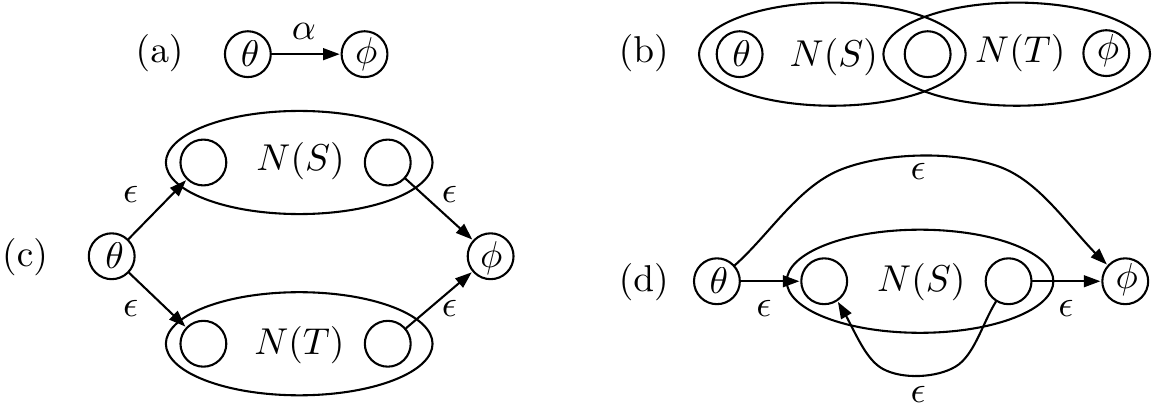}
    \caption{Thompson's recursive NFA construction. The regular
   expression $\alpha \in \Sigma \cup \{\varepsilon\}$ corresponds to NFA
   $(a)$. If $S$ and $T$ are regular expressions then $N(ST)$,
   $N(S|T)$, and $N(S^*)$ correspond to NFAs $(b)$, $(c)$, and $(d)$,
   respectively.  In each of these figures, the leftmost node $\theta$
   and rightmost node $\phi$ are the start and the accept nodes,
   respectively. For the top recursive calls, these are the start and
   accept nodes of the overall automaton.}
    \label{fig:Thompson}
\end{figure}

\subsection{Thompson Automata}
Given a regular expression $R$, we can construct an NFA accepting precisely the strings in the language of $R$ by several classical  methods~\cite{MY1960, Glushkov1961, Thompson1968}. In particular, Thompson~\cite{Thompson1968} gave the simple textbook construction shown in Figure~\ref{fig:Thompson}. We will call an NFA constructed with these rules a \emph{Thompson NFA} (TNFA). A TNFA $N(R)$ for $R$ has at most $2m$ states, at most $4m$ transitions, and can be computed in $O(m)$ time, where $m$ is the length of $R$. Furthermore, a TNFA has only one start state, $\theta$, and one accepting state, $\phi$. We use subscripts $\theta_A$ and $\phi_A$ to specify the start and accepting states of automaton $A$ when needed. We can compute a state-set transition on a single character in $O(m)$ time with a breadth-first search. Thus, we can solve regular expression matching on $Q$ in $O(nm)$ time and $O(m)$ space using the state-set transition algorithm above. This bound can be slightly improved by polylogarithmic factors, but by conditional lower bounds we should not hope to solve regular expression matching in  $O((nm)^{1-\epsilon})$ time for $\epsilon > 0$, see e.g., \cite{Myers1992, BFC2008, Bille2006, BT2009, BT2010, BI2016, BGL2017, Schepper2020}.

\subsection{Suffix Trees}


The \emph{suffix tree} $ST$ of a string $Q$ over an alphabet $\Sigma$ of length $n$ is the compacted trie over the set of all suffixes of $Q$. Each leaf is represented by the starting index in $Q$ of the suffix it represents. For a node $v$, we define $str(v)$ to be the string consisting of the concatenated labels on the path from the root to $v$. The \emph{string depth} of $v$ is the length of $str(v)$. The \emph{locus} of a string $q$ in $ST$ is the node $v$ in $ST$ with the smallest string depth such that $q$ is a prefix of $str(v)$. We can construct the suffix tree in $sort(n,|\Sigma|)$ time, where $sort(n,u)$ is the time to sort $n$ integers of an alphabet of size $u$~\cite{FCFM2000} and represent it in $O(n)$ space. Given the locus of a string $q$ in $ST$, we can output all starting indices of occurrences of $q$ in $Q$ by traversing the subtree below the locus of $q$ in $O(|\occ_q|)$ time, where $|\occ_q|$ is the number of occurrences of $q$ in $Q$.

\section{Interval Trees}\label{sec:intervaltree}

In this section, we introduce a data structure which we use to match the patterns throughout the paper.

Let $e$ be a regular expression of length $m$ with corresponding TNFA $A$ and let $Q$ be a string of length $n$ (we assume without loss of generality that $n$ is a power of 2).
The \emph{interval tree} $T$ of $Q$ for $e$ is a complete binary tree over $Q$ such that the nodes of $T$ correspond to the dyadic intervals of $Q$. In many cases we will omit mentioning the string, and simply refer to the interval tree for a regular expression. In these cases the string is the input string, $Q$.
The leaves in $T$ correspond to intervals of length 1 in $Q$ and  an internal node $v$ of height $k$ corresponds to an interval of length $2^k$. See Figure~\ref{fig:interval_tree}(a) for an example. 

Let  $v$ be a node in $T$. The index of the left and right ends of the interval of $v$ are denoted $l(v)$ and $r(v)$, respectively, and $c(v)$ denotes the 
middle of the interval rounded down.
For an integer $i \in [l(v), c(v)]$ and an integer $j \in [c(v) +1, r(v)]$ we define the sets of  \emph{left match states} and \emph{right match states}
\begin{alignat*}{2}
    \lms(v,i) &= \{s \in A \mid \text{there is a path from $\theta$ to $s$ matching $Q[i+1,c(v)]$}\} \\
    \rms(v,j) &= \{s \in A \mid \text{there is a path from $s$ to $\phi$ matching $Q[c(v)+1, j-1]$}\} 
\end{alignat*}
See Figure~\ref{fig:lms_rms} for an illustration of this.
Note that if $i = c(v)$ then $\lms(v,i)$ is the set of states $s$  such that there is a path from $\theta$ to $s$ matching the empty string, and similarly for $\rms(v,c(v)+1)$.

Each node $v$ in the interval tree stores the following information:
\begin{itemize}
    \item The index of the left and right ends of the interval, $l(v)$ and $r(v)$,  and $c(v)$, where $c(v)=\lfloor(l(v)+r(v))/2\rfloor$.
    \item The \emph{left match states} of $v$ for each  $i \in [l(v), c(v)]$, and the \emph{right match states} of $v$ for each $j \in[c(v)+1, r(v)]$. The sets are stored as bit vectors with an entry for each state in $A$, which has a 1 if the state is in the set, and a 0 otherwise.
\end{itemize}
\begin{figure}[t]
    \centering
    \includegraphics[width=1\linewidth]{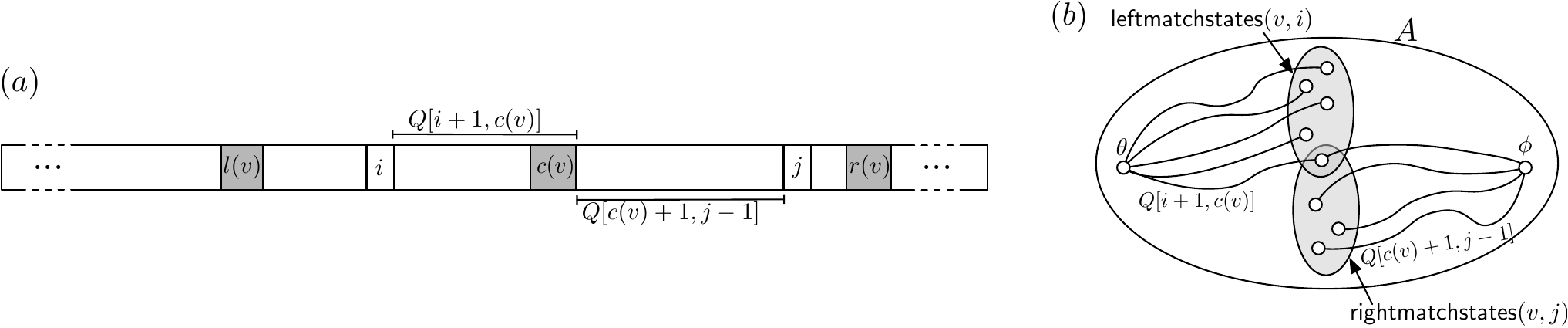}
    \caption{$(a)$ The substrings of a dyadic interval that are matched in the automaton $A$ to find $\lms(v,i)$ and $\rms(v,j)$ for some choice of $i$ and $j$. $(b)$ The left match states of $v$ from $i$ are the states in $A$ that can be reached from $\theta$ by a path that matches $Q[i+1,c(v)]$. The right match states of $v$ from $j$ are the states from which $\phi$ can be reached by a path matching $Q[c(v)+1,j-1]$.}
    \label{fig:lms_rms}
\end{figure}

 \begin{lemma}\label{lem:intervaltree}
    Given string $Q$ of length $n$ and regular expression $e$ of length $m$, the interval tree $T$ of $Q$ for $e$ can be computed in $O(n^2m)$ time and stored in $O(nm)$ space.
\end{lemma}

\begin{proof}
    We first construct the binary tree together with the information $r(v),l(v)$ and $c(v)$ for each node $v$. To compute the left and right match sets for each node we first construct the TNFA $A$ and its reverse $\overleftarrow{A}$. 
    
    An ancestor $u$ of a node $v$ is a \emph{left ancestor} of $v$ if $v$ is in the right subtree of $u$. A \emph{right ancestor} is defined symmetrically.     
    For each index $i$ in $Q$, we compute $\lms(v,i)$ for all nodes $v$ that are right ancestors of $Q[i]$ using state-set transitions on $A$. To do so, we start matching $Q[i+1,n-1]$ in $A$ and whenever we have matched $Q[i+1,c(v)]$ for a right ancestor $v$ of $Q[i]$, we store the current state-set of $A$ as $\lms(v,i)$.
    The sets $\rms(v,j)$ can be computed similarly using  state-set transitions on $\overleftarrow{A}$ and the reverse string $Q^R$. 
    
    We first analyze the time. The binary tree with $r(v),l(v)$ and $c(v)$ for each node can be computed in $O(n)$ time. Constructing $A$ and $\overleftarrow{A}$ takes $O(m)$ time.
    Using the state-set algorithm from every index in $Q$ takes $O(nm)$ time for each index. Thus the total time to compute all left and right match sets is $O(n^2m)$. In total the time to construct $T$ is $O(n^2m)$.
     
    

    We now analyze the space.
    Each node $v$ in the tree stores $l(v),r(v)$, and $c(v)$ in constant space per node.
    Let $|A|$ be the number of states in $A$.
    A node $v$ also stores the set $\lms(v,i)$ for each $i$ where $l(v)\leq i \leq c(v)$ and the set $\rms(v,j)$ for each $j$ where $c(v)+1\leq j \leq r(v)$.
    The bit vectors for the left and right match sets each use $O(|A|)$ bits, and as $|A|=O(m)$ such a bit vector uses $O(m/w)$ space.
    Multiplying by the size of the interval, the bit vectors for a node $v$ use $O(m/w\cdot (r(v)-l(v)))$ space.
    On each level of the tree this sums to $O(nm/w)$.
    Over all levels of the tree, this sums to $O(\frac{nm}{w} \log n)=O(nm)$.
\end{proof}

\subsection{Interval Tree Traversal}\label{sec:interval_tree_traversal}
We introduce a procedure to traverse the interval tree of a string $Q$ for a regular expression $e$. 
This procedure takes as input an interval tree and two sets of indices of $Q$, $X$ and $Y$, and determines for each $y\in Y$ whether there is a substring $Q[x+1,y-1]$ of $Q$ such that $x\in X$ and $Q[x+1,y-1]$ matches $e$. The procedure then returns the set $Z$ of all such indices $y$.
First, we introduce some concepts relating $X$ and $Y$ to the nodes of the interval tree for $e$.
For a node $v$ in the interval tree, we define $X(v)$ to be the indices $x$ of $X$ that are within the interval of $v$, that is, $X(v)=\{ x\in X~|~ l(v) \leq x \leq r(v)\}$. Likewise, $Y(v)=\{ y\in Y~|~ l(v) \leq y \leq r(v)\}$.
We also introduce $X_l(v)$ and $X_r(v)$, which are simply the indices of $X(v)$ partitioned such that $X_l(v)$ contains the indices of $X(v)$ that are in the left half of the interval, that is, $x\in X(v)$  where $x\leq c(v)$, while $X_r(v)$ contains the indices of $X(v)$ such that $x>c(v)$. We define $Y_l(v)$ and $Y_r(v)$ in the same way.
Another key element is the sets of left and right match states of a node $v$, which we call $\lms(v)$ and define as follows. The set $\lms(v)$ is the union of $\lms(v,x)$ for all $x$ in $X_l(v)$. That is, a state $s$ is in $\lms(v)$ if and only if there is an index $x$ in $X_l(v)$ such that there exists a path from $\theta$ to $s$ in the automaton for $e$ that matches $Q[x+1,c(v)]$.
The intuition behind $\lms(v)$ is that if $\lms(v)$ and $\rms(v,y)$ overlap for some $y\in Y$, then there exists an index $x$ in $X_l(v)$ such that there is a path in the TNFA for $e$ from $\theta$ to $\phi$ that matches $Q[x+1,y-1]$. 

\begin{figure}[t]
    \centering
    \includegraphics[width=1\linewidth]{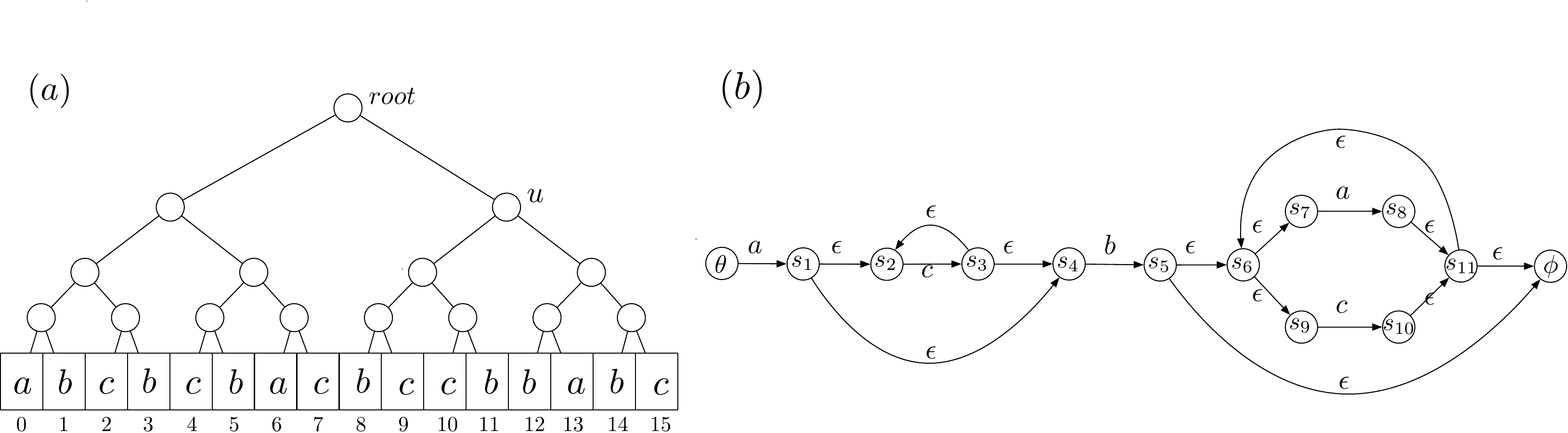}
    \caption{$(a)$ The interval tree of the string $Q = abcbcbacbccbbabc$ over the alphabet $\Sigma=\{a,b,c\}$. Here $l(u)=8$, $r(u)=15$ and $c(u)=11$. $(b)$ The TNFA for the regular expression $e_1=a (c^*)b(a|c)^*$. For the interval tree in $(a)$, $\lms(\rootnode,5)=\{s_2,s_3,s_4\}$ and $\rms(\rootnode,11)=\{s_1,s_3,s_4\}$}
    \label{fig:interval_tree}
\end{figure}

\paragraph{Traversal Procedure} With these definitions in place, we describe the traversal algorithm $\inttraversal$.
Let $T$ be the interval tree for $e$. Given the sets $X$ and $Y$, we traverse $T$ in the following way:

\paragraph{$\inttraversal(T,X,Y)$:} 
We traverse $T$ starting in the root. In the root of the interval tree, we have $X(\rootnode)=X$ and $Y(\rootnode)=Y$.
In a node $v$ of $T$, we perform the following 
steps:

\begin{description}
                \item[Step 1.] We partition $X(v)$ into $X_l(v)$ and $X_r(v)$, as well as $Y(v)$ into $Y_l(v)$ and $Y_r(v)$.
                \item[Step 2.] We compute 
                $\lms(v) = \bigcup_{x \in X_l(v)}\lms(v,x)$
                by $\mathsf{OR}$-ing the bitvectors.
                
                \item[Step 3.] For every $y \in Y_r(v)$ we compute $\lms(v) \cap \rms(v,y)$ by $\mathsf{AND}$-ing the bitvectors. 
                \item[Step 4.] We continue the traversal of $T$. Let $u$ and $w$ be the left and right child of $v$, respectively.
                If neither $X_l(v)$ nor $Y_l(v)$ is empty, recurse on $u$ with $X(u) = X_l(v)$ and $Y(u) = Y_l(v)$.
                Likewise if neither $X_r(v)$ nor $Y_r(v)$ is empty, recurse on $w$ with $X(w) = X_r(v)$ and $Y(w) = Y_r(v)$. If $u$ and $w$ are leaves, we do not recurse on them.
            \end{description}
After the traversal is complete, return $Z$.

\begin{example}\label{ex:inttraversal}
    Let  $Q = abcbcbacbccbbabc$ and $e=a(c^*)b(a|c)^*$ as represented by the interval tree and automaton in Figure~\ref{fig:interval_tree}. Furthermore, let $X = [3,5,8,11,12]$ and $Y = [8,11,12]$.
    Then, $X_l(\rootnode)=[3,5]$ and $Y_r(\rootnode)=[8,11,12]$.
    For the root node $\rootnode$ in the interval tree, $\lms(\rootnode) = \{s_2,s_3,s_4\}$ since $\lms(\rootnode,3)= \emptyset$ and $\lms(\rootnode,5)= \{s_2,s_3,s_4\}$. 
    For the right match states, 
    $\rms(\rootnode,8)= \{s_5,s_8,s_{10},s_{11},\phi\}$, $\rms(\rootnode,11)= \{s_1,s_3,s_4\}$ and $\rms(\rootnode,12)= \emptyset$.
    The only index $j$ for which $\rms(\rootnode,j)$ overlaps with $\lms(\rootnode)$ is $j=11$, where $\rms(\rootnode,j)\cap \lms(\rootnode)=\{s_3,s_4\}$, so $Z=\{11\}$. In this case $Q[6,10]$ matches $e_1$, $5\in X(\rootnode)$ and $11\in Y(\rootnode)$.
\end{example}

\noindent We show the correctness of the traversal procedure by proving the following lemma.

\begin{lemma}\label{lem:inttraversal_correctness}
    Given an interval tree $T$ of a string $Q$ for a regular expression $e$, and two sets of indices of $Q$, $X$ and $Y$, the output set $Z$ of $\inttraversal(T,X,Y)$ contains exactly the indices $y\in Y$ such that $Q[x+1,y-1]$ matches $e$ for some $x\in X$.
\end{lemma}
\begin{proof}
    We first show that for $y\in Z$, the substring $Q[x+1,y-1]$ matches $e$ for some $x\in X$. Since $y\in Z$, we know that for some $v$ in $T$, the intersection of $\lms(v)$ and $\rms(v,y)$ is non-empty.
    Consider a state $s$ in this overlap. By the construction of $\lms(v)$, $s$ is in $\lms(v)$ if and only if there is an index $x\in X_l(v)$ such that there is a path in the TNFA $A$ for $e$ from $\theta$ to $s$ matching $Q[x+1,c(v)]$. Likewise, $s$ is in $\rms(v,y)$ if and only if there is a path in $A$ from $s$ to $\phi$ matching $Q[c(v)+1,y-1]$. Thus, there is a path in $A$ from $\theta$ to $\phi$ matching $Q[x+1,y-1]$ for some $x\in X$.
    
    We now show that, given a substring $Q[x+1,y-1]$ that matches $e$ where $x\in X$ and $y\in Y$, by the end of the traversal $y\in Z$. Let $v$ be the lowest common ancestor of $x$ and $y$ in $T$. Then the interval of $v$ is the smallest one containing both $x$ and $y$, so $x \leq c(v) <y$, so $x\in X_l(v)$ and $y \in Y_r(v)$. If $Q[x+1,y-1]$ matches $e$, then there is a path $p$ from $\theta$ to $\phi$ in $A$ that matches $Q[x+1,y-1]$. The edges of $p$ can be partitioned into a path $p_1$ from $\theta$ to some state $s\in A$ such that $p_1$ matches $Q[x+1,c(v)]$, and another path $p_2$ from $s$ to $\phi$ such that $p_2$ matches $Q[c(v)+1,y-1]$. By the definitions, $s$ is in $\lms(v)$, and in $\rms(v,y)$, so $y \in Z$.
\end{proof}

\noindent We now consider the running time of the interval tree traversal.

\begin{lemma}\label{lem:inttraversal_runtime}
    Given an interval tree $T$ of a string $Q$ of length $n$ for a regular expression $e$ of length $m$, and two sets of indices of $Q$, $X$ and $Y$, $\inttraversal(T,X,Y)$ runs in $O((1+m/w)\sum_{v\in T}(|X(v)|+|Y(v)|))$ time and requires $O(n+m/w+1)$ extra space.
\end{lemma}
\begin{proof}
    We analyze the time spent on a node $v$ of $T$. We first partition the sets $X(v)$ and $Y(v)$ by comparing each index to $c(v)$. This takes $O(|X(v)|+|Y(v)|)$ time. We then compute $|X(v)|$ unions of state-sets using $\mathsf{OR}$-operations on bitvectors, in $O(|X(v)|(m/w+1))$ time. Lastly, we compute $|Y(v)|$ intersections of state-sets using $\mathsf{AND}$-operations on bitvectors, in $O(|Y(v)|(m/w+1))$ time. In total the time spent on a node $v$ is $O((1+m/w)(|X(v)|+|Y(v)|))$. Over all nodes in the tree, we thus spend $O((1+m/w)\sum_{v\in T}(|X(v)|+|Y(v)|))$ time.

    We now analyze the space. During the traversal we only need extra space for the partitions of the sets $X$ and $Y$, however, once $X(v)$ has been partitioned into $X_l(v)$ and $X_r(v)$ we no longer need $X(v)$ and likewise for $Y(v)$, so the sets of indices only require $O(n)$ space.
    We also need space for $\lms(v)$, but once the intersections have been computed, we no longer need $\lms(v)$, so we only need $O(m/w)$ space for this. In total, the interval tree traversal requires $O(n+m/w+1)$ space.
\end{proof}

\section{Single Linear Rewbs with a Single Reference}\label{sec:simplealgorithm}
In this section, we consider expressions of the form  $R=e_0(e)_1e_1\backslash1~e_2$ with the language $L(R) =\{w_0 q w_1 q w_2|w_i \in L(e_i),q \in L(e)\}$, and present an algorithm to determine whether such an expression $R$ of length $m$ matches a string $Q$ of length $n$ using $O(n^2m + n^2\log n)$ time and $O(nm)$ space.

We will start with a few definitions which we will need for the algorithm.
First, we define the set of indices $\prefixset$, which contains the indices $i$ of $Q$ such that the prefix $Q[0,i]$ matches $e_0$.
Likewise, the set $\suffixset$ contains the indices $j$ of $Q$ such that the suffix $Q[j,n-1]$ matches~$e_2$. Note, that if $\varepsilon$ matches $e_0$, $-1$ is in $\prefixset$, and if $\varepsilon$ matches $e_2$, $n$ is in $\suffixset$.
The \emph{prefix match indices} for $q$ which we will denote $\prefixmatch_q$ are the indices $i$ such that the prefix $Q[0,i]$ of $Q$ matches $e_0\odot q$.
Likewise, we define the \emph{suffix match indices} for $q$ called $\suffixmatch_q$ which are the indices $j$ such that the suffix $Q[j,n-1]$ of $Q$ matches $q\odot e_2$.
See Figure~\ref{fig:simple_algorithm} for an illustration.

\begin{figure}[t]
    \centering
    \includegraphics[width=0.68\linewidth]{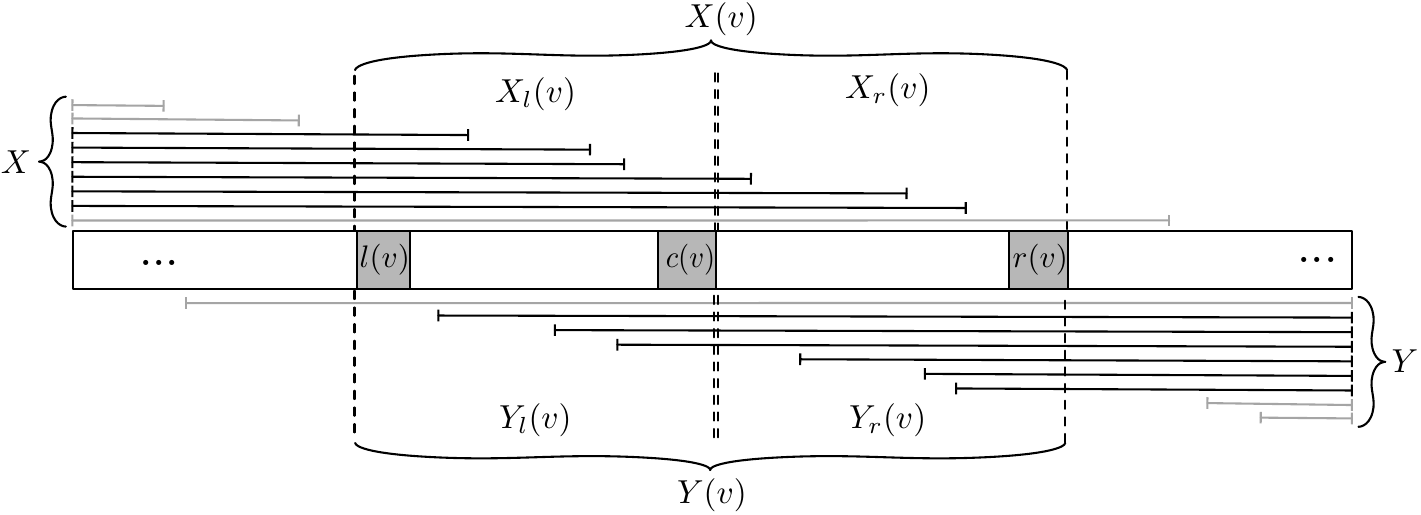}
    \caption{Illustration of the sets $X(v)$ and $Y(v)$ as well as the split into $X_l(v)$ and $X_r(v)$ and $Y_l(v)$ and $Y_r(v)$.}
    \label{fig:prefix_suffixmatch}
\end{figure}

\subsection{Algorithm}

The algorithm consists of three parts, preprocessing, finding substrings matching $e$, and processing matching substrings. 
In the preprocessing we construct the interval tree of $Q$ for $e_1$ and compute $\prefixset$ and $\suffixset$.
We then find the substrings $q$ of $Q$ that match $e$ and for each such $q$, we construct the list of occurrences of $q$ in $Q$. We find these substrings one at a time.
Once we have found a matching substring $q$ and the list of occurrences $\occ_q$ we process $q$ by traversing the interval tree as described in Section~\ref{sec:interval_tree_traversal} to determine whether $Q$ matches $e_0\odot q \odot e_1 \odot q\odot e_2$.
We then find and process the next matching substring  and list of occurrences, and repeat until we have processed all matching substrings of $e$.


\paragraph{Preprocessing}
We first construct the TNFAs $A_e$, $A_0$,  and $A_2$ for the regular expressions $e$, $e_0$, and $e_2$, respectively, and the reverse automaton $\overleftarrow{A_2}$.
We also construct the interval tree $T_1$ of $Q$ for $e_1$ and the suffix tree $ST$ of the string~$Q$. We then compute the sets of indices $\prefixset$ and $\suffixset$ using the TNFAs $A_0$ and $\overleftarrow{A_2}$.

We check if the empty string matches $e$ using $A_e$, and if so, we construct the TNFA for $e_0e_1e_2$, and use the state-set transition algorithm to check if $Q$ matches $R$ when the capturing group matches $\varepsilon$. If it does, we stop and return true.

\begin{example}\label{ex:preprocessing} 
 Consider the string from Figure~\ref{fig:interval_tree} and the regular expression $e_0 = (a|b|c)^*(b|c)$ describing the set of strings that end with $b$ or $c$, as well as $e_2= (b|c)^*a(b|c)^*$ describing the set of strings that contain exactly one $a$. Then $\prefixset = \{1,2,3,4,5,7,8,9,10,11,12,14,15\}$  and $\suffixset = \{ 7,8,9,10,11,12,13\}$.
\end{example}

\paragraph{Finding matching substrings}

In this part, we use the suffix tree to find the substrings $q$ of $Q$ that match $e$. 
We do this in the following way:

For a node $v$ in $ST$, let $S(v)$ be the set of states in the state-set simulation after matching $str(v)$ in $A_e$.
For the root $r$ of the suffix tree, $S(r)$ is the $\varepsilon$-closure of $\theta_{A_e}$.
We now traverse the suffix tree depth-first, starting in the root.
Given $S(v)$ for a node $v$, we process each child $u$ of $v$ in the following way:
Let $Q[i,j]$ be the edge label of the edge $(v,u)$.
Starting with the state-set $S(v)$, match $Q[i,j]$ in $A_e$ using state-set transitions.
For every $k$ where $i\leq k\leq j$ we check if $\phi_{A_e}$ is in the state-set after matching $Q[i,k]$. If so, then $q=str(v)\odot Q[i,k]$ matches $e$.
In this case, we construct the list of occurrences $occ_q$ of $q$ in $Q$ from the suffix tree, and process $q$ as described 
below. Otherwise, we continue. Once we have matched the entire label $Q[i,j]$, we save the current state-set as $S(u)$, and continue the traversal of the suffix tree.

\begin{example}\label{ex:occurrrences}
    Consider again the string from Figure~\ref{fig:interval_tree} and $e= ac|b$. Then the only two matching substrings are $ac$ and $b$, and $\occ_{ac} = [ 6]$ and $\occ_b=[1,3,5,8,11,12,14]$.
\end{example}

\paragraph{Processing a substring}\label{sec:processsubstring}

We compute the prefix and suffix match indices for $q$ by iterating through the indices $i$ in $\occ_q$, and letting $i-1+|q|\in\prefixmatch_q$ if and only if $i-1\in \prefixset$. Likewise, $i\in \suffixmatch_q$ if and only if $i+|q|\in \suffixset$. We then traverse the interval tree using the procedure $\inttraversal(T_1,\prefixmatch_q,\suffixmatch_q)$ 
as described in Section~\ref{sec:interval_tree_traversal}. If $Z$ is non-empty then $Q$ matches $R$.

\begin{example}\label{ex:presufmatch}
Consider the string and interval tree from Figure~\ref{fig:interval_tree}, as well as the sets $\prefixset$ and $\suffixset$ from Example~\ref{ex:preprocessing} and $\occ_b$ from Example~\ref{ex:occurrrences}.
Then, $\prefixmatch_b = [3,5,8,11,12]$ and $\suffixmatch_b =  [8,11,12]$.

\end{example}

\begin{figure}[t]
    \centering
    \includegraphics[width=0.75\linewidth]{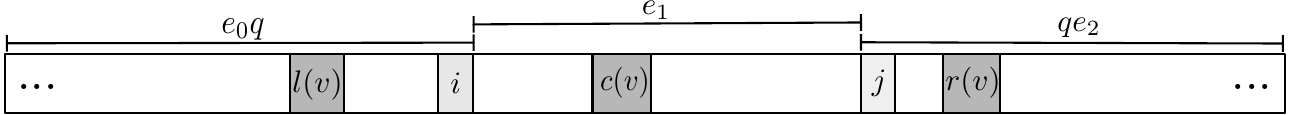}
    \caption{Illustration of a match. For some node $v$ in the interval tree, there is an $l(v) \leq i \leq c(v)$, such that $i$ is in $\prefixmatch_q$, meaning that $Q[0,i]$ matches $e_0 \odot q$, and an $c(v)+1 \leq j \leq r(v)$, such that $j$ is in $\suffixmatch_q$, meaning that $Q[j,n-1]$ matches $q \odot e_2$. Furthermore, $Q[i+1,j-1]$ matches $e_1$.}
    \label{fig:simple_algorithm}
\end{figure}

\subsection{Analysis}\label{sec:simpleanalysis}
We will first analyze the algorithm with respect to time.
Let $m_i$ be the length of the regular expression $e_i$ and $m_e$ be the length of $e$, such that $m=m_0+m_e+m_1+1+m_2$.

\paragraph{Preprocessing} The TNFAs $A_e$, $A_0$, $A_2$  and $\overleftarrow{A_2}$ can be built in $O(m_e + m_0+m_2)$ time. The sets $\prefixset$ and $\suffixset$ are computed in $O(n(m_0+m_2))$ time. 
The suffix tree of $Q$ can be built in $O(n\log n)$ time.
By Lemma~\ref{lem:intervaltree} we can build the interval tree in $O(n^2m_1)$ time. 
Checking if the empty string matches $e$ takes $O(m_e)$ time, constructing the automaton for $e_0e_1e_2$ takes $O(m_0+m_1+m_2)$ time, and checking if $Q$ matches $e_0e_1e_2$ takes $O(n(m_0+m_1+m_2))$.
Thus we use $O(n^2m_1+n(m_0+m_1+m_2) +m_e)$ time in total on preprocessing.

\paragraph{Finding matching substrings} 

When traversing and matching in the suffix tree, each label of the suffix tree is read once and the characters are matched in the TNFA for $e$, which takes $O(m_e)$ per character. 
The combined length of the labels in the suffix tree is $O(n^2)$, so reading and matching all labels takes $O(n^2m_e)$ time.

The combined length of $\occ_q$ over all substrings $q$ of $Q$ that match $e$ is $O(n^2)$ since there are $O(n^2)$ substrings of $Q$. Since we have the position of the locus of $q$ in the suffix tree the list $\occ_q$ can be computed in $O(|\occ_q|)$.  Thus the 
combined time to output $\occ_q$ for all $q$ is $O(n^2)$.
In total we use $O(n^2m_e)$ time on this part.

\paragraph{Processing substrings} 
First, we compute the sets $\prefixmatch_q$ and $\suffixmatch_q$ from $\occ_q$ and $\prefixset$ and $\suffixset$. The total length of the lists $\occ_q$ for all substrings $q$ of $Q$ is $O(n^2)$ so the combined time for computing $\prefixmatch_q$ and $\suffixmatch_q$ for all $q$ is $O(n^2)$ time.

Let $\occ_q(a,b)$ be the list containing the indices $i$ from $\occ_q$ where $a\leq i \leq b$, and define $\prefixmatch_q(v)$ and $\suffixmatch_q(v)$ to be the indices within the interval of $v$, that is, $l(v) \leq i \leq r(v)$.
We now look at the time for the interval tree traversals. For this, we will use the following lemma:
\begin{lemma}\label{lem:suffixmatch_size}
    For a node $v$ in the interval tree, 
    \begin{itemize}
        \item $\sum_{q\in Q}|\suffixmatch_q(v)|\leq \sum_{q\in Q}|\occ_q(l(v),r(v))| \leq (r(v)-l(v)+1)\cdot n$
        \item $\sum_{q\in Q}|\prefixmatch_q(v)|\leq \sum_{q \in Q}|\occ_q(l(v)-|q|+1,r(v)-|q|+1)| \leq (r(v)-l(v)+1)\cdot n$
    \end{itemize}
    
\end{lemma}
\begin{proof}
    We know that $|\suffixmatch_q(v)|\leq|\occ_q(l(v),r(v))|$ since every suffix $Q[i,n-1]$ where $i\in \suffixmatch_q$ starts with an occurrence of $q$.
    Likewise, $|\prefixmatch_q(v)|\leq|\occ_q(l(v)-|q|+1,r(v)-|q|+1)|$ since $\occ_q(l(v)-|q|+1,r(v)-|q|+1)$ contains the indices of occurrences of $q$ that end in the interval of $v$.
    Lastly, we note that at most $n$ different substrings can start on the same index, so $\sum_{q\in Q}|\occ_q(a,b)|\leq (b-a+1)\cdot n$, which leads to the lemma.
\end{proof}

From Lemma~\ref{lem:inttraversal_runtime} we know that $\inttraversal(T,\prefixmatch_q,\suffixmatch_q)$ runs in $O((1+m_1/w)~\sum_{v\in T}(|\prefixmatch_q(v)|+|\suffixmatch_q(v)|))$ time. For all substrings $q$ of $Q$ that match $e$, we therefore have $O((1+m_1/w)~\sum_{v\in T}\sum_{q\in Q}(|\prefixmatch_q(v)|+|\suffixmatch_q(v)|))$. Using Lemma~\ref{lem:suffixmatch_size} we get $O((1+m_1/w)~\sum_{v\in T}(r(v)-l(v))n)$. Since the combined size of the intervals of all nodes in the tree is $n\log n$, we get $O((1+m_1/w)n^2\log n)=O(n^2m_1+n^2\log n)$.

In total, we use $O(n^2\log n + n^2(m_1+m_e)+n(m_0+m_2))=O(n^2m+n^2\log n)$ time.
We now analyze the algorithm's space usage.

\paragraph{Preprocessing} The TNFAs $A_e$, $A_0$, $A_2$, and $\overleftarrow{A_2}$ use $O(m_e + m_0+m_2)$ space, the sets $\prefixset$ and $\suffixset$ use $O(n)$ space, and the interval tree uses $O(nm_1)$ space by Lemma~\ref{lem:intervaltree}. The suffix tree uses $O(n)$ space, and the automaton for $e_0e_1e_2$ uses $O(m_0+m_1+m_2)$ space. This gives us $O(nm_1+m_e + m_0+m_2)$ space.

\paragraph{Finding matching substrings} 

The sets $S(v)$ use $O(m_e)$ space each, and we store one such set per node in the suffix tree, so $O(nm_e)$ space in total. 
We only need $\occ_q$ for one substring at a time, and the size of each $\occ_q$ is $O(n)$.
The space complexity of this part then becomes $O(nm_e)$.

\paragraph{Processing substrings} The suffix and prefix match indices of $q$ use $O(n)$ space, and by Lemma~\ref{lem:inttraversal_runtime} the traversal uses $O(n+m/w+1)$ extra space.
\medskip

\noindent In total, the algorithm uses $O(n(m_1+m_e)+m_0+m_2)=O(nm)$ space.

\subsection{Correctness}

Given a substring $q$ of $Q$ as well as the interval tree $T_1$ of $e_1$, we use the procedure $\inttraversal(T_1,\prefixmatch_q,\suffixmatch_q)$ to check if $Q$ matches $e_0\odot q\odot e_1\odot q\odot e_2$. If the output set $Z$ is not empty, then there exist indices $i$ in $\prefixmatch_q$ and $j$ in $\suffixmatch_q$ such that $e_1$ matches $Q[i+1,j-1]$. Since $Q[0,i]$ matches $e_0\odot q$, $Q[i+1,j-1]$ matches $e_1$ and $Q[j,n-1]$ matches $q\odot e_2$, then $Q$ matches $e_0\odot q\odot e_1\odot q\odot e_2$. 
If $Z$ is empty, then there are no indices $i$ in $\prefixmatch_q$ and $j$ in $\suffixmatch_q$ such that $e_1$ matches $Q[i+1,j-1]$, and we can conclude that $Q$ does not match $e_0\odot q\odot e_1\odot q\odot e_2$.
See Figure~\ref{fig:simple_algorithm} for an illustration.
We traverse the interval tree for every non-empty substring $q$ of $Q$ that matches $e$, and explicitly check the case where the capturing group matches the empty string, so the correctness of the algorithm follows from the correctness of $\inttraversal$.
In summary, we have shown the following lemma.






\begin{lemma}
Given a single linear rewb of length $m$ with one reference and a string of length $n$, we can solve the regular expression matching with backreferences problem in $O(n^2m+n^2 \log n)$ time and $O(nm)$ space.
\end{lemma}

\section{Single Linear rewbs}\label{sec:generalization1}
In this section we consider single linear rewbs, i.e.,  of the form $R=e_0(e)_1e_1\backslash1~e_2 \ldots e_{k}\backslash1~e_{k+1}$. 
We extend our techniques from Section~\ref{sec:simplealgorithm} to this case and solve the single linear matching problem in $O(n^2m +n^2 k\log n)$ time and $O(nm)$ space.

The algorithm is similar to the one from Section~\ref{sec:simplealgorithm}, except we now traverse $k$ interval trees $T_1\ldots T_{k}$ for the regular expressions $e_1\ldots e_{k}$.
Let $Z_h$ be the output of the interval tree traversal for $T_h$.
While in Section~\ref{sec:simplealgorithm} we only needed to know whether the output set of $\inttraversal$ was empty, in this section we use the indices of $Z_h$ to find the prefix match indices for the traversal of $T_{h+1}$, $\prefixmatch_{q,h+1}$. Those are the positions $i$ where $Q[0,i]$ matches $e_0\odot q \odot e_1 \odot q \odot \ldots  \odot q \odot e_{h} \odot q$.
See Figure~\ref{fig:multiple_backreferences}.

\subsection{Algorithm}

Let $\prefixset$ be the indices $i$ such that $Q[0,i]$ matches $e_0$ as before, but now, let $\suffixset$ be the indices $j$ such that $Q[j,n-1]$ matches $e_{k+1}$.

%


\paragraph{Preprocessing}

We build the TNFAs $A_e$, $A_0$ and $\overleftarrow{A_{k+1}}$ and use $A_0$ and $\overleftarrow{A_{k+1}}$ to compute $\prefixset$ and $\suffixset$.
Then we construct the interval trees $T_1\ldots T_{k}$.
We check if the empty string matches $e$ using $A_e$, and if so, we construct the TNFA for $e_0e_1e_2\ldots e_{k+1}$, and use the state-set transition algorithm to check if $Q$ matches $R$ when the capturing group matches $\varepsilon$. If it does, we stop and return true.

\paragraph{Finding matching substrings}
Finding matching substrings is done just as in Section~\ref{sec:simplealgorithm}, and we process each matching substring $q$ as described below. 
Furthermore, we initialize a set $\mathsf{endpoints}$ to which we will add the indices in $\prefixmatch_{q,k+1}$ after the traversal of $T_k$ for a substring $q$.

\paragraph{Processing a substring}

We compute $\prefixmatch_{q,1}$ as before by iterating through the indices $i$ in $\occ_q$, and letting $i-1+|q|\in\prefixmatch_{q,1}$ if and only if $i-1\in \prefixset$.

We then traverse the $k$ interval trees $T_1\ldots T_k$ as follows. For $h=1\ldots k$, we compute $Z_h = \inttraversal( T_h, \prefixmatch_{q,h}, occ_q)$. We then compute $\prefixmatch_{q,h+1}=\{y-1+|q|, y\in Z_h\}$, and traverse $T_{h+1}$ (unless $h=k$). See Figure~\ref{fig:multiple_backreferences}.
After traversing $T_1 \ldots T_{k}$ for a substring $q$, we add the indices of $\prefixmatch_{q,k+1}$ to the set $\mathsf{endpoints}$.

\begin{figure}[t]
    \centering
    \includegraphics[width=0.8\linewidth]{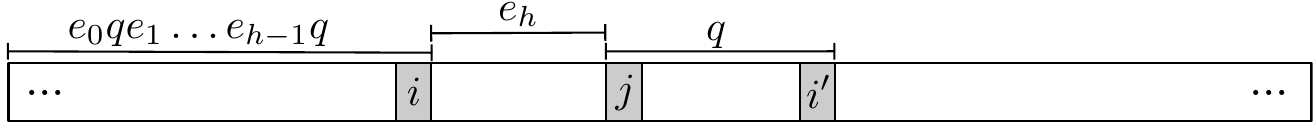}
    \caption{In the traversal of $T_h$, we find indices $j$, such that $Q[i+1,j-1]$ matches $e_h$, and $i\in \prefixmatch_{q,h}$. Then, $i'=j-1+|q|\in \prefixmatch_{q,h+1}$.}
    \label{fig:multiple_backreferences}
\end{figure}

\paragraph{Postprocessing}
After processing all matching substrings, we check for every index $i$ in $\mathsf{endpoints}$ whether $i+1$ is in $\suffixset$. If such an $i$ exists, $R$ matches $Q$. Otherwise, $R$ does not match $Q$.

\subsection{Analysis}
Again, let $m_i$ be the length of the regular expression $e_i$ and $m_e$ be the length of $e$, such that $m=m_0+m_e+m_1+1+m_2+1+\ldots +1+m_{k+1}$. We first analyze the 
running time.

\paragraph{Preprocessing}
As in Section~\ref{sec:simpleanalysis}, we can build $A_e$, $A_0$, $A_{k+1}$ and $\overleftarrow{A_{k+1}}$ in $O(m_e+m_0+m_{k+1})$ time and compute $\prefixset$ and $\suffixset$ in $O(n(m_0+m_{k+1}))$ time.
Building the interval tree of $Q$ for $e_h$ takes $O(n^2m_h)$ according to Lemma~\ref{lem:intervaltree}, so building the interval trees $T_1\ldots T_{k}$ takes $O(n^2\sum_{h=1}^{k}m_h)$. 
Checking if the empty string matches $e$ takes $O(m_e)$ time, and checking if $Q$ matches $e_0e_1e_2\ldots e_{k+1}$ takes $O(n(\sum_{h=0}^{k+1}m_h))$.
This gives us $O(n^2(\sum_{h=1}^{k}m_h)+n(\sum_{h=0}^{k+1}m_h))+m_e)$.

\paragraph{Finding matching substrings}
The only change from the simpler algorithm in Section~\ref{sec:simplealgorithm} is the addition of the set $\mathsf{endpoints}$. We analyze only this, as the rest is unchanged, and thus takes $O(n^2m_e)$.
The total number of times an index is added to the set $\mathsf{endpoints}$ is the combined length of $\prefixmatch_{q,k+1}$ over all matching substrings $q$. Since $|\prefixmatch_{q,k+1}| \leq |\occ_q|$, the time spent is  $O(\sum_{q\in Q} |\prefixmatch_{q,k+1}|)\leq O(\sum_{q\in Q} |\occ_q|) = O(n^2)$.

\paragraph{Processing a substring}
Computing $\prefixmatch_{q,1}$ from $occ_q$ for all substrings $q$ takes $O(n^2)$, by the same argument as in Section~\ref{sec:simpleanalysis}.
For the interval tree traversals, we consider the traversals of $T_h$ for every substring $q$. By the same argument as in Section~\ref{sec:simpleanalysis}, the traversals of $T_h$ in total take $O(n^2m_h+n^2 \log n)$. For all $k$ interval trees, we get $O(\sum_{h=1}^k n^2m_h+n^2 k\log n)$.

In total the algorithm uses $O(n^2 k\log n+n^2(m_e+\sum_{h=1}^{k} m_h)+n(m_0+m_{k+1}))=O(n^2 m+n^2k \log n)$ time.
\medskip

We now consider the space. We use $O(n(\sum_{h=1}^{k}m_h))$ space for the interval trees and $O(m_e+\sum_{h=0}^{k+1}m_h)$ space for the automata, while still using $O(nm_e)$ for finding matching substrings of $e$. 
Prefix match indices use $O(n)$, and $occ_q$ uses $O(n)$ space and is only needed for one substring at a time. The size of $\mathsf{endpoints}$ is also $O(n)$.
Thus, in total we use $O(n(m_e+\sum_{h=1}^{k}m_h)+m_0+m_{k+1})=O(nm)$ space.

\subsection{Correctness}\label{sec:generalization1_correctness}
We first show the following lemma.
\begin{lemma}\label{lem:prefixmatch_k}
    $\prefixmatch_{q,k+1}$ contains the positions $i$ where $Q[0,i]$ matches $e_0\odot q \odot e_1 \odot q \odot \ldots  \odot q \odot e_{k} \odot q$.
\end{lemma}
\begin{proof}
    We prove the lemma by induction.
Due to the correctness of the state-set transition algorithm, $\prefixmatch_{q,1}$ contains $i$ if and only if $Q[0,i]$ matches $e_0\odot q$. 
Assume that $\prefixmatch_{q,h}$ contains $x$ if and only if $Q[0,x]$ matches $e_0\odot q \odot e_1 \odot q \odot \ldots  \odot q \odot e_{h-1} \odot q$. Then, by the previous assumption and Lemma~\ref{lem:inttraversal_correctness}, $Z_h = \inttraversal( T_h, \prefixmatch_{q,h}, occ_q)$ contains the indices $y$ such that $Q[0,y-1]$ matches $e_0\odot q \odot e_1 \odot q \odot \ldots  \odot q \odot e_{h}$ and $Q[y,y+|q|-1]$ is an occurrence of $q$. Then, $\prefixmatch_{q,h+1}=\{y-1+|q|, y\in Z_h\}$ contains $x$ if and only if $Q[0,x]$ matches $e_0\odot q \odot e_1 \odot q \odot \ldots  \odot q \odot e_{h} \odot q$. 
\end{proof}

By Lemma~\ref{lem:prefixmatch_k}, the set $\mathsf{endpoints}$ contains all indices $i$ such that $Q[0,i]$ matches $e_0\odot q \odot e_1 \odot q \odot \ldots  \odot q \odot e_{k} \odot q$ for all substrings $q$ that match $e$.
Thus, $\mathsf{endpoints}$ contains all indices $i'$ such that $Q[0,i']$ matches $e_0(e)_1e_1\backslash1 ~ e_2\backslash1 ~ e_3 \ldots  e_{k} \backslash1$. Now, if there is an index $i'$ in $\mathsf{endpoints}$ such  that $i'+1$ is in $\suffixset$, then $R=e_0(e)_1e_1\backslash1~e_2\backslash1~e_3\ldots e_{k}\backslash1~e_{k+1}$ matches $Q$. 
In summary, we have shown the following lemma:

\begin{lemma}
    Given a single linear rewb of length $m$ with $k$ references and a string $Q$ of length $n$, we can solve the single linear matching problem in $O(n^2m+n^2k\log n)$ time and $O(nm)$ space.
\end{lemma}

\section{Ordered Linear Rewbs}\label{sec:generalization2}
In this section we extend our results to solve the matching problem for ordered linear rewbs in $O(n^2m + n^2k \log n)$ time, where $k$ is the number of references, and $O(nm)$ space. 
These rewbs can be easily rewritten by moving the boundaries between the concatenated single linear rewbs, to fit the form $R=E_0 E_1\ldots E_t e_{t+1,0}$, where $E_\ell$ is a subexpression of the form $E_\ell=e_{\ell,0}(d_\ell)_\ell ~ e_{\ell,1}~\backslash \ell ~e_{\ell,2}\backslash \ell\ldots  \backslash \ell ~ e_{\ell,k_\ell}\backslash \ell$ with $e_{\ell,j}$ and $d_\ell$ for $0 \leq \ell \leq t+1$ and $0\leq j \leq k_\ell$ being regular expressions.


\subsection{Algorithm}
This type of rewb is essentially a sequence of rewbs of the type we consider in Section~\ref{sec:generalization1}, and we essentially apply the algorithm of Section~\ref{sec:generalization1} repeatedly.
For each $E_\ell$, we find matching substrings of $d_\ell$ as in Section~\ref{sec:generalization1}, and for each matching substring we traverse the $k_\ell$ corresponding interval trees $T_{\ell,1}\ldots T_{\ell,k_\ell}$. In this way, we find all indices $i$ such that $Q[0,i]$ matches $E_0\ldots E_\ell$. These indices are then used as prefix match indices for $E_{\ell+1}$.

\paragraph{Preprocessing}
    Let $\prefixset_0$ be the indices $i$ such that $Q[0,i]$ matches $e_{0,0}$, and $\suffixset$ be the indices $j$ such that $Q[j,n-1]$ matches $e_{t+1,0}$. We build the TNFAs $A_{0,0}$ and $\overleftarrow{A_{{t+1},0}}$ for $e_{0,0}$ and $e_{t+1,0}$ and use them to compute $\prefixset_0$ and $\suffixset$.

 \paragraph{Processing $E_\ell$ for $\ell=0\ldots t$}
    Starting with $E_0$, we process each subexpression sequentially as follows: First, we build the interval trees $T_{\ell,0}\ldots T_{\ell,k_\ell}$ of $Q$ for $e_{\ell,0}\ldots e_{\ell,k_\ell}$ (except for $T_{0,0}$ which we will not need).  We then construct $A_{d_\ell}$. Using the state-set transition algorithm on $A_{d_\ell}$, we check whether $d_\ell$ matches the empty string $\varepsilon$. If so, we construct the TNFA $A_\varepsilon$ for $e_{\ell,0}e_{\ell,1}e_{\ell,2}\ldots e_{\ell,k_\ell}$. For $\ell=0$ we use the state-set transition algorithm on $A_\varepsilon$ from the beginning of $Q$, and for $\ell>0$, we use the state-set transition algorithm on $A_\varepsilon$ from $i+1$ for every index $i$ in $\mathsf{endpoints}_{\ell-1}$. This way, we find all indices matching $E_0\ldots E_\ell$ when the capturing group matches $\varepsilon$. We add these indices to $\mathsf{endpoints}_\ell$.

    For $\ell>0$, we then compute $\prefixset_\ell$ by Using the state-set transition algorithm on the TNFA for $e_{\ell,0}$ from $i+1$ for every index $i$ in $\mathsf{endpoints}_{\ell-1}$ to find every index $i'$  such that $Q[0,i']$ matches $E_0\ldots E_{\ell-1}e_{\ell,0}$.
    
    We then find all matching substrings of $d_\ell$ one at a time by using the same method as in Section~\ref{sec:generalization1}. Given a substring $q$ that matches $d_\ell$, we traverse the $k_\ell+1$ interval trees for $e_{\ell,0}\ldots e_{\ell,k_\ell}$. Let $T_{\ell,h}$ be the interval tree for $e_{\ell,h}$, and $Z_{\ell,h-1}$ be the output set from the interval tree traversal of $T_{\ell,h-1}$

    To traverse $T_{\ell,h}$ we do as follows.
    First, we compute $\prefixmatch_{q,\ell,h}$. We do this according to the following 3 cases:

    \begin{description}
        \item[Case 1: $h=0$.] We do not compute $\prefixmatch_{q,\ell,h}$ or traverse an interval tree in this case.
        \item[Case 2: $h=1$.] We compute $\prefixmatch_{q,\ell,h}$ from $\prefixset_\ell$ and $\occ_q$ by iterating through the indices $i$ in $\occ_q$, and letting $i-1+|q|\in\prefixmatch_{q,\ell,h}$ if and only if $i-1\in \prefixset_\ell$.
        \item[Case 3: Otherwise.]  We compute $\prefixmatch_{q,\ell,h}=\{y-1+|q|, y\in Z_{\ell,h-1}\}$
    \end{description}
    
    We then traverse $T_{\ell,h}$ to compute $Z_{\ell,h}=\inttraversal(T_{\ell,h},\prefixmatch_{q,\ell,h},occ_q)$.
    
    After traversing $T_{\ell,1} \ldots T_{\ell,k_\ell}$ for a substring $q$, we compute $\prefixmatch_{q,\ell,h+1}$ and add the indices to the set $\mathsf{endpoints_{\ell}}$. We then find the next matching substring and repeat. Once we have traversed $T_{\ell,1} \ldots T_{\ell,k_\ell}$ for all substrings $q$ that match $d_\ell$, $\mathsf{endpoints_{\ell}}$ contains all indices $i$ for which $Q[0,i]$ matches $E_0\ldots E_\ell$.

\begin{figure}
    \centering
    \includegraphics[width=0.8\linewidth]{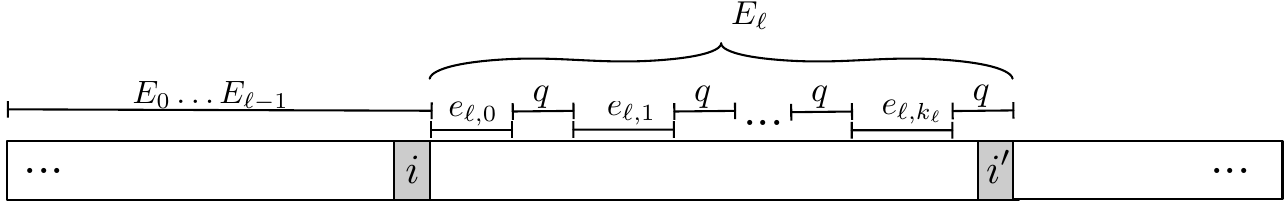}
    \caption{For an index $i'$ to be in $\mathsf{endpoints_{\ell}}$, there is an index $i$ in $\mathsf{endpoints_{\ell-1}}$ and a substring $q$ that matches $d_\ell$ such that $Q[i+1,i']$ matches 
    $e_{\ell,0}\odot q \odot e_{\ell,1}\odot \ldots \odot q \odot e_{\ell,k_\ell}\odot q$.}
    \label{fig:multiple_capturinggroups}
\end{figure}

\paragraph{Postprocessing}
    After processing $E_t$, and computing $\mathsf{endpoints}_t$, we check for every index $i$ in $\mathsf{endpoints}_t$ whether $i+1$ is in $\suffixset$. If such an $i$ exists, $R$ matches $Q$. Otherwise, $R$ does not match $Q$.

\subsection{Analysis}
Let $m_{\ell,h}$ be the length of the regular expression $e_{\ell,h}$, and $m_{d_\ell}$ be the length of the regular expression $d_\ell$ such that
\[
    m=\sum_{\ell=0}^{t} (m_{\ell,0}+m_{d_\ell}+m_{\ell,1}+1+\ldots +1+m_{\ell,k_\ell})+m_{t+1,0}
    =\sum_{\ell=0}^{t} \big(m_{d_\ell}+O(k_\ell)+\sum_{h=0}^{k_\ell} m_{\ell,h}\big)+m_{t+1,0}
\]

We first analyze the time. 
Building the TNFAs $A_0$ and $\overleftarrow{A_{{t+1},0}}$ takes $O(m_{0,0}+m_{t+1,0})$ time, and using them to compute $\prefixset_0$ and $\suffixset$ takes $O(n(m_{0,0}+m_{t+1,0}))$.
We build interval trees $T_{\ell,h}$ for $e_{\ell,h}$ for $\ell=0\ldots t$ and $h=1\ldots k_\ell$. By Lemma~\ref{lem:intervaltree}, it takes $O(n^2m_{\ell,h})$ to build $T_{\ell,h}$, so for all interval trees, we spend $O(n^2(\sum_{\ell=0}^t(\sum_{h=1}^{k_\ell} m_{\ell,h})-m_{0,0}))$ time. 
Constructing $A_{d_\ell}$ and checking if the empty string matches $d_\ell$ takes $O(nm_{d_\ell})$ time. Constructing the TNFA for  $e_{\ell,0}e_{\ell,1}e_{\ell,2}\ldots e_{\ell,k_\ell}$ takes $O(\sum_{h=0}^{k_\ell}m_{\ell,h})$ time, and using the state-set transition algorithm on it from every index in $\mathsf{endpoints}_{\ell-1}$ takes $O(n^2(\sum_{h=0}^{k_\ell}m_{\ell,h}))$. 
Constructing $\prefixset_\ell$ takes $O(n^2m_{\ell,0})$ time, so for $\ell=0\ldots t$, we use $O(n^2(\sum_{\ell=0}^{t}m_{\ell,0}))$ time.

Finding all matching substrings for each $d_\ell$ takes $O(n^2(m_{d_\ell})$ by the same argument as in Section~\ref{sec:generalization1}, so $O(n^2\sum_{\ell=0}^t m_{d_\ell})$ in total.

We compute the prefix match indices according to three cases. In case 1, no additional work is done. In cases 2 and 3, we iterate through at most $|occ_q|$ indices, which sums to $O(n^2)$ for all $q$, and we do this at most once for each of the $k$ interval trees, for a total of $O(n^2k)$ time.

The traversal of the interval trees takes $O(n^2m_{\ell,h}+n^2\log n)$ for $T_{\ell,h}$ by Lemma~\ref{lem:inttraversal_runtime}. Recall, that we do not traverse an interval tree for $e_{\ell,0}$. Traversing all interval trees then takes $O(n^2k\log n+n^2(\sum_{\ell=0}^t (\sum_{h=1}^{k_\ell} m_{\ell,h})-m_{0,0}))$. In total, the algorithm uses
\begin{align*}
    &O(n^2k\log n + n^2\sum_{\ell=0}^t m_{d_\ell}+n^2(\sum_{\ell=0}^t(\sum_{h=1}^{k_\ell} m_{\ell,h})-m_{0,0})+n(m_{0,0}+m_{t+1,0})) \\
    =&O(n^2k\log n + n^2\big(\sum_{\ell=0}^t \big(m_{d_\ell}+\sum_{h=1}^{k_\ell} m_{\ell,h}\big)\big)+n(m_{0,0}+m_{t+1,0}))\\
    =&O(n^2k\log n + n^2m)
\end{align*}
time.

We now analyze the space. The $A_0$ and $\overleftarrow{A_{{t+1},0}}$ take $O(m_{0,0}+m_{t+1,0})$ space, and the interval trees $T_{\ell,h}$ for $e_{\ell,h}$ for $\ell=0\ldots t$ and $h=0\ldots k_\ell$ except for $e_{0,0}$ use $O(nm_{\ell,h})$ by Lemma~\ref{lem:intervaltree} for $T_{\ell,h}$ so for all interval trees, we use $O(n(\sum_{\ell=0}^t(\sum_{h=0}^{k_\ell} m_{\ell,h})-m_{0,0}))$ space. The TNFA $A_{d_\ell}$ uses $O(m_{d_\ell})$, and the TNFA for $e_{\ell,0}e_{\ell,1}e_{\ell,2}\ldots e_{\ell,k_\ell}$ uses $O(\sum_{h=0}^{k_\ell}m_{\ell,h})$ space.
Finding matching substrings for each $d_\ell$ takes $O(n(m_{d_\ell})$ by the same argument as in Section~\ref{sec:generalization1}.
The total space for the algorithm is then 
\begin{align*}
    &O(n\sum_{\ell=0}^t m_{d_\ell}+n(\sum_{\ell=0}^t(\sum_{h=0}^{k_\ell} m_{\ell,h})-m_{0,0})+m_{0,0}+m_{t+1,0}) \\
    =&O(n\big(\sum_{\ell=0}^t \big(m_{d_\ell}+\sum_{h=0}^{k_\ell} m_{\ell,h}\big)\big)+(m_{0,0}+m_{t+1,0}))\\
    =&O(nm).
\end{align*}

\subsection{Correctness}

\begin{lemma}\label{lem:endpoints}
    For $1\leq l \leq t$, suppose that $\mathsf{endpoints}_{\ell-1}$ contains exactly the indices $i$ such that $Q[0,i]$ matches $E_0\ldots E_{\ell-1}$, then $\mathsf{endpoints}_\ell$ contains exactly the indices $i'$ such that $Q[0,i']$ matches $E_0\ldots E_\ell$.
\end{lemma}
\begin{proof}
    By the assumption and the correctness of the state-set transition algorithm, $\prefixset_\ell$ contains exactly the indices $i$ such that $Q[0,i]$ matches $E_0\ldots E_{\ell-1}e_{\ell,0}$.
    Consider a non-empty substring $q$ that matches $d_\ell$. Filtering its occurrences using $\prefixset_\ell$ gives exactly the indices $i$ where $Q[0,i]$ matching $E_0\ldots E_{\ell-1}e_{\ell,0}\odot q$ ends.

    Applying the same induction as in Lemma~\ref{lem:prefixmatch_k} to the traversals of $T_{\ell,1},\ldots,T_{\ell,k_\ell}$, we obtain exactly the indices $i'$ such that $Q[0,i']$ matches $E_0\ldots E_{\ell-1}e_{\ell,0}\odot q\odot e_{\ell,1}\odot q\odot\ldots\odot e_{\ell,k_\ell}\odot q$. These indices are added to $\mathsf{endpoints}_{\ell}$. This is done for every substring $q$ that matches $d_\ell$, proving the lemma.

\end{proof}

For $\ell=0$, the argument from Section~\ref{sec:generalization1_correctness} shows that $\mathsf{endpoints}_{0}$ contains exactly the the indices $i$ such that $Q[0,i]$ matches $E_0$. Repeatedly applying Lemma 9 then shows the corresponding property for $\mathsf{endpoints}_{t}$. Finally, $Q$ matches $R$ if and only if some $i\in\mathsf{endpoints}_{t}$ satisfies $i+1\in\suffixset$. Together with the analysis in Section 6.2, this proves Theorem~\ref{thm:ordered_linear}.

\bibliography{references}

\end{document}